\documentclass[draftclsnofoot,onecolumn, 12pt]{IEEEtran}

\usepackage{color}
\usepackage{amssymb}
\usepackage{amsmath}
\usepackage{amsthm}
\newcommand{\subparagraph}{}
\usepackage{titlesec}
\usepackage{graphicx}
\usepackage{graphicx}
\usepackage{url}
\usepackage{caption}
\usepackage{subcaption}
\usepackage{hyperref}
\usepackage{geometry}
\usepackage{setspace}
\usepackage{authblk}
\usepackage{csquotes}
\newcommand{\bb}[1]{\mathbf{#1}}

\newcommand{\mm}[1]{\mathrm{#1}}

\makeatletter

\newcommand{\Rmnum}[1]{\expandafter\@slowromancap\romannumeral #1@}

\makeatother

\newtheorem{theorem}{Theorem}

\begin{document}
\title{On the Capacity  of Massive MIMO  With 1-Bit ADCs and DACs at the Receiver and  at the Transmitter}
\author{Armin~Bazrafkan
       and Nikola~Zlatanov\thanks{Armin Bazrafkan and Nikola Zlatanov are with the Department of Electrical and Computer Systems Engineering, Monash University, Melbourne, VIC 3800, Australia (emails: armin.bazrafkan@monash.edu and nikola.zlatanov@monash.edu)}}

\date{}
\maketitle
\begin{abstract}
In this paper, we investigate the capacity  of massive  multiple-input multiple-output (MIMO) systems corrupted by complex-valued additive white Gaussian noise (AWGN) when both the transmitter and the receiver employ  1-bit digital-to-analog  converters (DACs) and 1-bit analog-to-digital-converters (ADCs).  As  a result of 1-bit DACs and ADCs, the transmitted and received symbols, as well as  the transmit- and receive-side  channel state information (CSI),  are   assumed to be quantized to 1-bit of information. The derived results show that the capacity of the considered massive MIMO system is  $2N$ and $2M$ bits per channel use     when $N$ is fixed and $M\to\infty$ and when $M$ is fixed and $N\to\infty$, respectively,
where $M$ and $N$ denote  the number of transmit and receive antennas, respectively.   These coincide with the respective capacities with full CSI at both the transmitter and the receiver. 
In both cases,  we showed that the derived capacities can be achieved with noisy 1-bit CSI at the transmitter-end or at the receiver-end,  and without any CSI at the other end.
  Moreover, we showed that the capacity can be achieved  in one channel use without employing channel coding, which results in a latency of one channel use.
\end{abstract}

 \section{Introduction}
Vast available spectrum resources have made millimeter-wave (mm-wave) communications a key technology for the development of 5G wireless systems \cite{ref47}. Experimental measurements show that mm-wave communication links can provide multi-gigabit rates with low latency \cite{ref48}. However, mm-wave signals suffer from large propagation loss and are more prone to blockage than  conventional micro-wave signals \cite{ref6}. In order to mitigate the high path-loss of mm-wave signals, large antenna arrays at the transmitter and at the receiver are necessary for directional beamforming  \cite{ref7}. For example, Huawei has already launched its massive MIMO technology in China that  adopts large scale antenna arrays that are able to move both vertically and horizontally for 3D beamforming. This technology is able to achieve 1.4 Gbps in the 3.5 GHz band by adopting 40 MHz channels \cite{ref49}.  

Fortunately, due to the small wavelength of mm-wave signals and advances in radio frequency (RF) circuits, hundreds of antennas can be placed in a few square centimeters, resulting in mm-wave massive multiple-input multiple-output (MIMO) \cite{ref8}, \cite{ref9}. Among the main design problems for the massive MIMO technology is the  power consumption of analog-to-digital converters (ADCs) and digital-to-analog converters (DACs) associated with each RF chain.  Specifically,  due to   large bandwidths in mm-wave MIMO, the sampling rates of the ADCs and DACs have to increase, which results in  high-resolution, high-speed ADCs and DACs that consume huge amounts of energy and are costly \cite{ref11}, \cite{ref12},  \cite{ref13}. Since for standard deployments an RF chain is connected to each antenna, this issue limits the number of antennas at  the transmitter and the receiver. 
\subsection{Low-Resolution DACs and ADCs}
Generally, there are two approaches to reduce the power consumption of ADCs/DACs. With the first approach,  the sampling rate of the ADCs and DACs is reduced by making a number of parallel low-speed converters  act as a high-speed one \cite{ref50}. However, this typically results in a mismatch in gain, timing, and voltage among the sub-levels of the parallel ADCs/DACs, which leads to error floors in link  performance \cite{ref11}, \cite{ref14}, \cite{ref15}.  In the second approach, instead of reducing the sampling rate of the  ADCs/DACs, their resolutions are reduced. Since the power consumption of ADCs/DACs grows exponentially with the resolution \cite{ref15}, ADCs/DACs with resolution of 1-bit have considerably lower power consumption compared to typical high-resolution ADCs/DACs with 10-12 bits \cite{ref14}. Furthermore,  since 1-bit ADCs/DACs do not need automatic gain control,  they can also significantly reduce the system complexity  \cite{ref11}, \cite{ref18}.  However,  due to the low resolution of these ADCs/DACs, drastic non-linearities are added  to the system; and consequently,  the shapes of the transmitted/received waveforms cannot be preserved at the transmitter/receiver \cite{ref17}. Therefore, it is important to investigate how  1-bit ADCs/DACs  influence the capacity of massive MIMO systems.  Motivated by this,  in this paper, we investigate the capacity of the additive white Gaussian noise (AWGN) massive  MIMO systems  with  1-bit ADCs and DACs at the receiver and  at the transmitter.
\subsection{Related Works}
  One of the earliest works on the problem of finding the capacity of MIMO systems, where low resolution ADCs is \cite{ref18}. In \cite{ref18}, the capacity of a MIMO channel with 1-bit ADCs at the receiver is studied at low signal-to-noise ratios (SNRs) under the assumption of     perfect   channel state information (CSI) at the receiver  (CSIR). This problem shows that, at low SNRs,  the  capacity is smaller  by a factor of $2/\pi$, which corresponds to a gap of -1.96 dB in $E_b/N_0$, as compared to a system with infinite resolution.   In \cite{ref21}, the results in \cite{ref18} have been extended to the case where the additive Gaussian noise is mutually correlated across the receive antennas. Bussgang decomposition has been used in \cite{ref21} to model the corresponding MIMO channel and it has been demonstrated that certain conditions on the  channel and the noise covariance matrices result in lower performance loss compared to the case of uncorrelated noise.  In \cite{ref44} and \cite{ref45}, achievable rates have been presented for MIMO systems with 1-bit ADCs at the transmitter and full-precision DACs at the transmitter, where the antenna outputs are processed by analog combiners, and full CSI is available both at the transmitter and the receiver. 

In \cite{ref23}, it has been shown that at high SNRs,  under the assumption of perfect CSIT and CSIR, the capacity of MIMO systems with high-resolution inputs and 1-bit quantized outputs   is lower bounded by the rank of the channel matrix. On the other hand,  in \cite{ref11}, when the channel matrix has full-row rank, a tight upper bound on the capacity has been derived for the MIMO systems at finite SNRs. In \cite{ref11}, the results of \cite{ref23} have been extended, where under the assumption of full CSIT and CSIR, the  capacity of  MIMO  systems with 1-bit quantization at the receiver only has  been derived for infinite SNR. In \cite{ref24}, an achievable rate  has been derived for a massive MIMO system with high resolution DACs at the transmitter and  1-bit ADCs at the receiver under the assumption of imperfect CSIT, where the wideband frequency-selective channel is  being estimated by a linear low-complexity algorithm.  The authors in \cite{ref25} investigated the achievable  rate of a massive MIMO system where low resolution ADCs are used at the receiver only and perfect CSIT is available, where it has been shown that the performance loss caused by using low resolution ADCs can be compensated by increasing the number of antennas at the receiver.

In \cite{ref19}, the authors analyze the mutual information of a MIMO system with high resolution DACs at the transmitter and 1-bit ADCs at the receiver,  where   CSI is not available  at the transmitter (CSIT) nor at the receiver. Achievable rates in \cite{ref19} are provided only for the quantized single-input single-output (SISO) channels, where on-off QPSK signaling is shown to be capacity achieving. In \cite{ref20}, the authors analyze the mutual information of a MIMO system with high-resolution DACs at the transmitter and 1-bit ADCs at the receiver   is derived in the low SNR regime under the assumption of  no CSIT/CSIR. Achievable rates  in \cite{ref20} are provided only for the SIMO channel and only in the asymptotic case of low SNRs.

\subsection{Main Contributions}
In   prior works, the assumption of  perfect CSI has been leveraged for deriving the achievable rates and/or capacity results for MIMO systems\footnote{The achievable rate in \cite{ref19} is for SISO and in \cite{ref20} for SIMO, but in the asymptotic low-SNR regime only.}. Although different algorithms have been proposed for estimating the channel while 1-bit ADCs are used at the receiver, they are associated with large estimation  errors, high computational complexity, and require extremely long training sequences \cite{ref33}\nocite{ref34}\nocite{ref35}\nocite{ref36}\nocite{ref37}\nocite{ref38}\nocite{ref39}-\cite{ref40}.
 Therefore, the assumption of  having full CSI when 1-bit ADCs/DACs are present, is unrealistic.  To the best of our knowledge, this is the first work on MIMO systems with 1-bit DACs/ADCs at the transmitter and at the receiver, that adopts the practical assumption of  noisy   1-bit   CSI  at the transmitter and/or at the receiver. 
 
 In this paper, we investigate two system models of massive MIMO, one when the number of  transmit antennas, denoted by $M$, is very large and goes to infinity, and the second when the number of receive antennas, denoted by $N$, is very large and goes to infinity. Next, the results are derived for  complex-valued AWGN MIMO channels.  The derived results show that the capacity of the considered massive MIMO system is  $2N$ and $2M$ bits per channel use     when $N$ is fixed and $M\to\infty$ and when $M$ is fixed and $M\to\infty$ hold, respectively. These coincide with the respective capacities with full CSI at both transmitter and receiver. 
In both cases,  we showed that the derived capacities can be achieved with noisy 1-bit CSI at the transmitter-end or at the receiver-end,  and without any CSI at the other end.
  Moreover, we showed that the capacity can be achieved  in one channel use without employing channel coding, which results in a latency of one channel use. Therefore, massive MIMO systems with 1-bit DACs/ADCs maybe a practical approach for achieving ultra reliable low latency communication (URLLC).

This paper is organized as follows. In Section \Rmnum {2}, the system model of a MIMO system with 1-bit quantization at both the transmitter and the receiver is described. In Section \Rmnum 3, the capacities of  massive MIMO systems with 1-bit quantization at both the transmitter and the receiver are presented.   Finally, conclusion is brought  in Section \Rmnum 4.


\section{System Model}
 We consider a MIMO system comprised of $M$ transmit and $N$ receive antennas. Each antenna element is equipped with two\footnote{The transmitter requires 1-bit DAC and 1-bit ADC   in order to transmit information symbols  and    receive pilot symbols in a Time Division Duplex (TDD) manner, respectively. In contrast, the receiver requires 1-bit ADC and 1-bit DAC   in order to receive information symbols  and    transmit pilot symbols in a TDD manner, respectively.} 1-bit quantizers, 1-bit ADC and  1-bit DAC,  that  quantize the received and transmitted signals, respectively, as shown in Fig. \ref{fig.1}. Let $\bb x\in \mathcal X^M$ denote  the complex-valued $M \times 1$ transmit vector after the 1-bit quantization at the transmitter. We assume that each element of $\bb x$ has unit energy. Let  $P$ be the total transmit power and  let $\bb y$ denote the  $N\times 1$ received vector before the 1-bit quantization at the receiver. Then,  $\bb y$ is given by  
\begin{equation}\label{eq.1}
  \bb y=\sqrt{\frac{P}{M}}\bb H \bb x+\bb w,
\end{equation}
where  $\bb H$ denotes the   $N\times M$ MIMO complex-valued channel matrix and $\bb w$ is the  $N \times 1$ complex-valued Gaussian noise vector with independent and identically distributed (i.i.d.) entries having zero mean and unit variance. Following standard Rayleigh fading, the channel matrix $\bb H$ is also assumed to have complex-valued Gaussian i.i.d. entries with zero mean and unit variance.  
The vector $\bb y$ undergoes 1-bit quantization at the receiver, yielding the $N\times 1$ quantized received vector $\bb z\in\{1+j,1-j,-1+j,-1-j\}^N$, given by
\begin{equation}\label{eq.2}
\bb z=\textrm {sign}(\bb y)=\textrm {sign}\left(\sqrt{\frac{P}{M}}\bb H\bb x+\bb w\right),
\end{equation}
 where 
 \begin{equation}\label{eq.2.1}
  \textrm{sign}(a+jb)=
 \left\{ \begin{array}{rl}
                     1+j & \mbox{if } a\geq 0 \textrm{ and } b\geq 0 \\
                     -1+j & \mbox{if }   a<0 \textrm{ and } b\geq 0\\
                      1-j & \mbox{if }   a\geq 0 \textrm{ and } b<0\\
                      -1-j & \mbox{if }   a< 0 \textrm{ and } b<0.\\
                   \end{array}
                   \right.
\end{equation}

\begin{figure}[t]
  \centering
  \includegraphics[width=15cm]{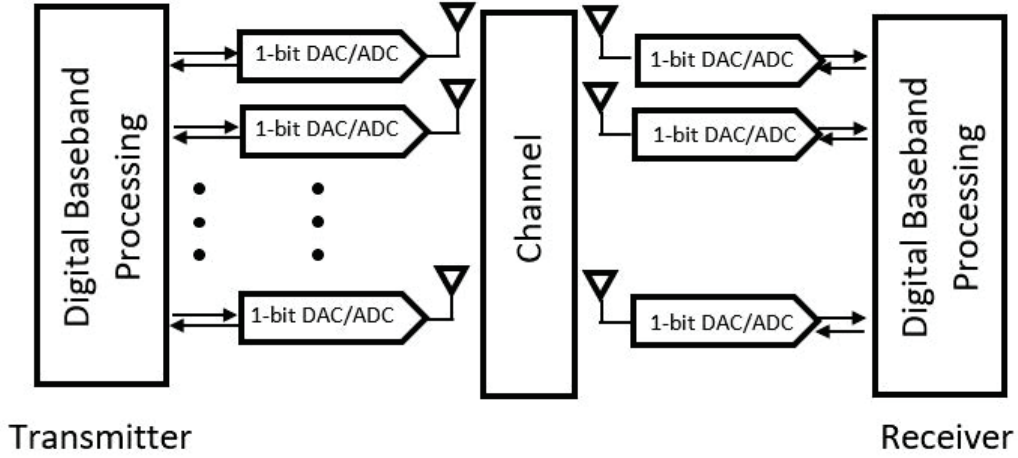}
  \caption{An $M \times N$ MIMO system with 1-bit quantization at the transmitter and the receiver}\label{fig.1}
\end{figure}

   Due to the 1-bit ADCs/DACs at both the receiver and the transmitter,  the transmitter and the receiver can have access only to 1-bit CSI corrupted by noise. More precisely, let $\bb G$ denote the noisy 1-bit estimate of the   channel matrix $\bb H$, which is  obtained by $K$ pilot transmissions per  antenna, given by 
 \begin{equation}\label{eq.3}
 \bb G=\textrm {sign}\left(\sum_{k=1}^K\textrm {sign}\left(  \sqrt{P_p}  \bb H\bb +\bb W_k\right)\right),
 \end{equation}
where $\bb W_k$ is an $N\times M$ noise matrix with i.i.d. zero-mean unit variance  complex-valued Gaussian elements, and   $P_p$ is the pilot power. In practice, assuming channel reciprocity, the noisy 1-bit CSI channel matrix $\mathbf G$ can be obtained at the transmitter (receiver)  by  sending $K$ orthogonal training symbols per antenna  from the receiver (transmitter) to the transmitter (receiver) in a TDD fashion. The transmitter (receiver) then collects the $K$ 1-bit quantized pilot signals received on each of its antennas,  then makes a 1-bit estimate on the channel based on majority rule and thereby obtains $\bb G$ in \eqref{eq.3}.

The  capacity of the considered MIMO system with 1-bit quantized inputs and outputs, and  noisy 1-bit CSI is given as
 \begin{equation}\label{eq.4}
   C=\underset{\mm p(\bb x|\bb G)}{\text{max}}\hspace{4mm}\mm I(\bb z;\bb x|\bb G),
 \end{equation}
 where $\mm p(\bb x|\bb G)$ is the  probability mass function (PMF) of  the transmit signal $\bb x\in\mathcal X^M$ given the available 1-bit noisy CSI matrix $\bb G$. In this paper, we derive the capacity in \eqref{eq.4} by focusing on the massive MIMO regime,  where the number of either transmit or receive antennas goes to infinity, i.e., either $M\to\infty$ and $N<\infty$ is fixed or $N\to\infty$ and $M<\infty$ is fixed. To this end, we assume that $P>0$ and $P_p>0$ hold.

\section{Capacity}
In this section, the capacities of massive MIMO systems with 1-bit quantized inputs and outputs,  and noisy 1-bit CSI  are presented for the  cases when the  massive antenna array is at the  transmit-side and at receive-side. 
\subsection{Massive Antenna Array At The Transmit-Side}
 The following theorem establishes the capacity when $M\to\infty$ and $N$ is fixed. 
\begin{theorem}\label{theo.4}
  The capacity of the  massive MIMO system with 1-bit quantized inputs and outputs,  and  noisy 1-bit CSI   satisfies the limit
  \begin{equation}\label{eq.cmimo_complex}
  \underset{M\rightarrow\infty}{\text{$\lim$}}\hspace{2.5mm}C=2N,
  \end{equation}
  where $N<\infty$ is fixed.
 The capacity can be achieved in one channel use by a scheme that uses the noisy 1-bit CSI available at the transmitter only, neglecting any CSI at the receiver. 
\end{theorem}
\begin{proof}
  See  Appendix \ref{app.1} for the proof.
\end{proof}
 We now demonstrate a scheme that is asymptotically able to communicate at this rate with vanishing probability of error by using  the  noisy 1-bit CSI only at the transmitter. 

\textbf{Achievability Scheme 1:} 
 In order to communicate $2N$ bits per channel use, the transmitter constructs a codebook comprised of $2^{2N}$ codewords $\bb s_i\in\mathcal S^{N}$, for $i=1,\cdots,2^{2N}$, where $\mathcal S= \{-1-j,-1+j,1-j,1+j\}$.   Without loss of generality, assume that   codeword $\bb s$ is selected to be transmitted in the considered channel use. Let $s_{n}^R\in\{-1,1\}$ and $s_{n}^I\in\{-1,1\}$ denote the  real-valued and imaginary-valued parts of the $n$-th complex-valued element of $\bb s$, respectively,  for $n=1,2,...,N$. Next, assume that $\bb G$, given by \eqref{eq.3}, is known at the transmitter. Let $g_{nm}^R\in\{-1,1\}$ and $ g^I_{nm}\in\{-1,1\}$ denote the real-valued and imaginary-valued parts of the $(n,m)$  element of   $\bb G$. 
   The transmit vector $\bb x\in\mathcal X^M$, where $\mathcal X=\{-1,1,-j,j\}$, is formed from   $\bb s$ and $\bb G$ as follows.  The vector $\bb x$  is divided into $N$ parts,  each  comprised of $M/N$ consecutive complex-valued symbols. Let $x_m^R\in\{-1,0,1\}$ and $x_m^I\in\{-1,0,1\}$ denote the  real-valued and imaginary-valued parts of the $m$-th complex-valued element of $\bb x$. Assume that $x_m^R$ and $x_m^I$ belong to the $n$-th group of transmit antennas,    for $n=1,2,...,N$. Then, $x_m^R$ and $x_m^I$  are constructed as
\begin{align}\label{aa1}
\begin{bmatrix} 
x^R_m \\
x^I_m 
\end{bmatrix}
=
 \frac{1}{2}   \begin{bmatrix} 
g^R_{nm} & g^I_{nm} \\
-g^I_{nm} & g^R_{nm} 
\end{bmatrix}
\begin{bmatrix} 
s^R_n \\
s^I_n 
\end{bmatrix} .
\end{align}
Note that $x_m^R$ and $x_m^I$, constructed using (\ref{aa1}), are such that either $x_m^R$ or $x_m^I$ is always zero while the other element is always $1$ or $-1$. Thereby, in a given channel use, the $m$-th transmit antenna is always silent on either the real-valued or the imaginary-valued channel/carrier.
Next,  $x_m^R$ and $x_m^I$ are amplified by $\sqrt{\frac{P}{M}}$, and then transmitted in one channel use from the $m$-th transmit antenna over the real-valued and imaginary-valued channels, respectively,   for $m=1,...,M$. The
  receiver receives $\bb z$ and it will then  decide that  $ \bb z$ has been the transmitted codeword in the given channel use. As a result, an error happens at the receiver if $\bb z\neq\bb s$ occurs. In Appendix \ref{app.1}, we prove that the error rate goes to zero as $M\to\infty$ for fixed $N<\infty$.

Achievability Scheme 1 works by splitting the $M\to\infty$ antennas at the transmitter into $N$ groups of antennas, where each group is comprised of $M/N\to\infty$ antennas. Then, the $n$-th group of transmit antennas, for $n=1,2,...,N$, uses the 1-bit CSI vector, obtained from the pilots by the $n$-th receive antenna, to beamform towards the direction of the $n$-th receive antenna. Thereby, for $M/N\to\infty$, the beamformed signal from the $n$-th group of transmit antennas is received amplified at the    $n$-th receive antenna and completely attenuated  at any other receive antennas. Hence, Achievability Scheme 1  resolves the   considered MIMO system   into $N$ parallel complex-valued Gaussian channels with 1-bit quantized inputs and outputs, where each parallel channel has an SNR that satisfies  $SNR\to \infty$  as $M/N\to\infty$. As a result, instantaneous error-free decoding is feasible.


\subsection{Massive Antenna Array At The Receive-Side}
We now turn to the opposite case when the massive antenna array is at the receiver's side.
The following theorem provides the capacity of this massive MIMO system for the asymptotic case when $N\to\infty$ holds and $M$ is fixed.
\begin{theorem}\label{theo.5}
   The capacity of a  massive MIMO system with 1-bit quantized inputs and outputs,  and   noisy 1-bit CSI   satisfies the limit
  \begin{equation}\label{eq.cmimo_c}
  \underset{N \rightarrow\infty}{\text{$\lim$}}\hspace{2.5mm}C=2M,
  \end{equation}
  where $M<\infty$ is fixed.
   The capacity can be achieved in one channel use by a scheme that  uses the noisy 1-bit CSI available at the receiver only, neglecting any CSI at the transmitter. 
\end{theorem}
\begin{proof}
  Please refer to Appendix \ref{app.2} for the proof. 
\end{proof}

\textbf{Achievability Scheme 2:}
 In order to communicate $2M$ bits per channel use, the transmitter constructs a codebook comprised of $2^{2M}$ codewords $\bb s_i\in\mathcal S^{M}$, for $i=1,\cdots,2^{2M}$, where $\mathcal S= \{-1-j,-1+j,1-j,1+j\}$.   Without loss of generality, assume that   codeword $\bb s$ is selected to be transmitted in the considered channel use. Then,  the transmit vector $\bb x\in\mathcal X^M$, where $\mathcal X=\frac{1}{\sqrt 2}\mathcal S$,   is constructed as
 \begin{align}
      \bb x =\frac{1}{\sqrt{2}}\bb s. 
 \end{align} 
 Next,  the elements of $\bb x$ are amplified by $\sqrt{P/M}$ and  then  are transmitted in one channel use from the $M$  transmit antenna. The
  receiver receives $\bb z$. Assume that $\bb G$ given by \eqref{eq.3} is known at the receiver. Then the receiver decides  that  
  \begin{align}\label{eq_dec-s}
      \bb{\hat s}= \bb G \bb z
  \end{align}
  has been the transmitted codeword.
   As a result, an error happens at the receiver if $\bb{\hat s}\neq\bb s$. 
   In Appendix \ref{app.2}, we prove that the error rate goes to zero as $N\to\infty$ when $M<\infty$ is fixed.

Achievability Scheme 2 works by the receiver using its $N\to\infty$ antennas to steer its reception  from the $M$  directions characterized by the $M$ column vectors of the 1-bit CSI matrix $\bb G$, respectively.  On a given direction, the receiver  receives  a complex-valued symbol  on each of its $N$ receive antennas, which can be either equal or   not equal to the actual transmit symbol. 
For    $N\to\infty$ and $M<\infty$ being fixed, the number of  received symbols which are equal to the actual transmit symbol on a given direction is always larger than the  number of received symbols which are not equal to the transmit symbol, leading to  an error-free transmission. 

 \subsection{Observations}
 
 In the following, we provide some observations regarding the derived results.
 
 Given the 1-bit quantizer at the  receiver, it is immediate to conclude that the capacity in \eqref{eq.cmimo_complex} is upper bounded by $2N$ (bits/channel use) independent of the level of quantization at the transmitter and of the CSI knowledge. In other words, this capacity result holds also when  full-precision symbols are transmitted, as well as when  full-precision CSI is present at both the transmitter and the receiver.
 Similarly, given the 1-bit quantizer at the   transmit side, it is immediate to conclude that the capacity in \eqref{eq.cmimo_c} is upper bounded by $2M$ (bits/channel use) independent of the level of quantization at the receiver and of the CSI knowledge. 
 
For the case when $M\to\infty$ and $N<\infty$ is fixed ($N\to\infty$ and $M<\infty$ is fixed), the number of  pilots symbols needed for 1-bit CSI estimation  in order to achieve the capacity in \eqref{eq.cmimo_complex} (in \eqref{eq.cmimo_c})  scales with the number of receive antennas $N$ (transmit antennas M).
 
 It is interesting to see that the capacity limits in Theorems~\ref{theo.4} and \ref{theo.5} hold for any fixed $P>0$ and $P_p>0$. Specifically, the relation between $P$ and $M$ in the error-rate of Achievability Scheme 1, given in \eqref{q7}, is of the form of $PM$. Hence, for any fixed $P>0$ and $M\to\infty$, $PM\to\infty$ holds and  the error-rate of Achievability Scheme~1 goes to zero as $PM\to\infty$. On the other hand,  the relation between $P_p$ and $M$ in the error-rate of Achievability Scheme 1, given in \eqref{q7}, can be written in the form of $M(1-2p_\epsilon)$, where $p_\epsilon$ is the probability of having an incorrect 1-bit CSI estimation on a single channel, which depends on the pilot power $P_p$. Thereby, $P_p$ via $p_\epsilon$ reduces the number of transmit antennas from $M$ to $M(1-2p_\epsilon)$. However, since $p_\epsilon<1/2$ holds for any fixed $P_p>0$, see \eqref{er2} for $K=1$, the effective number of transmit antennas $M(1-2p_\epsilon)$ still satisfies  $M(1-2p_\epsilon)\to\infty$, which in turn results in an error-rate that goes to zero as $M\to\infty$. Similar analysis holds also for Achievability Scheme~2.

Since each receive/transmit antenna receives/transmits its information without requiring  coordination with the other receive/transmit antennas,   Achievability Scheme  1/2  can also be applied to multi-user massive MIMO, i.e., to a MIMO system where the receive/transmit antennas are implemented on  non-cooperating distinct  devices.
 
Finally, it is interesting to note that Achievability Schemes 1 and 2 do not require channel coding  at the transmitter. As a result,  the latency that Achievability Schemes 1 and 2 achieve  is  one channel use. Hence, massive MIMO systems with 1-bit ADCs and 1-bit DACs maybe a practical approach for achieving URLLC.

\section{Conclusion}
In this paper, we presented the capacities  of the complex-valued AWGN massive MIMO system when the inputs,  outputs, and noisy CSI are quantized to one bit of information, for the cases when $M\to\infty$ and $N$ is fixed, and $N\to\infty$ and $M$ is fixed.  We showed that the capacity of the considered MIMO systems are  $2N$ and $2M$ (bits per channel use) when $M\to\infty$  and $N\to\infty$ hold, respectively.  In both cases, we showed that the capacity can be achieved  in one channel use without using channel coding, which results in a latency of one channel use. Moreover, the derived capacities can be achieved with noisy 1-bit CSI at the transmitter-end or at the receiver-end,  and without any CSI at the other end.
\begin{appendices}

\newpage

\section{Proof of Theorem \ref{theo.4}}\label{app.1}
\subsection{Converse}
 For  the considered MIMO system, we have
 \begin{equation}\label{eq.36.01}
   \mm{I}(\bb x;\bb z|\bb G)=\mm{H}(\bb z|\bb G)-\mm{H}(\bb z|\bb x,\bb G)\leq\mm{H}(\bb z|\bb G)\leq 2N,
 \end{equation}
 where the last inequality follows due to the 1-bit quantized outputs.
 Hence, the  capacity  of this system cannot be larger than $2N$. In the following, we prove that \eqref{eq.36.01} is asymptotically achievable when $M\to\infty$ and $N<\infty$ is fixed. 
 \subsection{Achievability}
 
Using  Achievability Scheme 1, the transmitter sends $2N$ bits of information at each channel use. Hence, the rate of Achievability Scheme 1 is $2N$ bits per channel use. Now, we are only left to prove that the  symbols received at the receiver can be decoded with  probability of error that vanishes   when $M\to\infty$ and $N<\infty$ is fixed.

Assume that $\bb s$ is transmitted and $\bb z$ received.  Then, the probability of error can be bounded as 
\begin{align} 
   \mm{P}_{\mm e}&= {\rm Pr}\{\bb z\neq \bb s\}\leq \sum_{n=1}^{N}\left(\mm{Pr}(z_{n}^R\neq  s_{n}^R)+\mm{Pr}(z_{n}^I\neq  s_{n}^I)\right) \label{eq.204a}\\
   &= 2N  \mm{Pr}\{z_{1}^R\neq  s_{1}^R\} \label{eq.204b},
\end{align}
where \eqref{eq.204a} follows from the union bound and \eqref{eq.204b} follows from symmetry, i.e., since $\mm{Pr}\{z_{n}^R\neq  s_{n}^R\}=\mm{Pr}\{z_{k}^R\neq  s_{k}^R\}=\mm{Pr}\{z_{n}^I\neq  s_{n}^I\}=\mm{Pr}\{z_{k}^I\neq  s_{k}^I\}$ holds. In the following, we derive a simplified expression for $\mm{Pr}\left\{z^R_1\neq s^R_1\right\}$. 

From \eqref{eq.2}, we can obtain    the real-valued  and imaginary-valued parts of the quantized received symbol at the first receive antenna, $z_{1}^R$ and $z_{1}^I$, as
\begin{align} 
    z_{1}^R&=\textrm{sign}\left(\sqrt{\frac{P}{M}}\sum_{m=1}^{M}\left(h_{1m}^Rx_{m}^R-h_{1m}^I x_{m}^I\right)+w_1^R\right) \label{aa1.1} ,\\
    z_{1}^I&=\textrm{sign}\left(\sqrt{\frac{P}{M}}\sum_{m=1}^{M}\left(h_{1m}^Rx_{m}^I+h_{1m}^I x_{m}^R\right)+w_1^I\right). \label{aa1.12}
\end{align}
Only  transmit symbols from transmit antennas $m=1,2,...,M/N$ are intended for the first receive antennas, and the symbols coming from all other antennas act as interference. Having this in mind,    (\ref{aa1.1}) and (\ref{aa1.12})    can be written as
\begin{align}
    z_{1}^R&=\textrm{sign}\left(\sqrt{\frac{P}{M}}\left( \sum_{m=1}^{M/N} (h_{1m}^R x_{m}^R-h_{1m}^I x_{m}^I) +v^R_1+w_1^R\right) \right)\label{aa1.1-1} ,\\
    z_{1}^I&=\textrm{sign}\left( \sqrt{\frac{P}{M}} \left(\sum_{m=1}^{M/N} (h_{1m}^R x_{m}^I+h_{1m}^I x_{m}^R)  +v^I_1 +w_1^I\right)\right), \label{aa1.12-1}
\end{align}
where 
\begin{align}
    v_{1}^R&= \sum_{m=M/N+1}^{M}\left(h_{1k}^R x_{k}^R-h_{1k}^I x_{k}^I\right)   \label{aab1.1} ,\\
    v_{1}^I&=  \sum_{m=M/N+1}^{M}\left(h_{1k}^R x_{k}^I+h_{1k}^I x_{k}^R\right)    \label{aab1.12}
\end{align}
are the interference at the $n$-th receive antenna. Since $v_{1}^R$ and  $v_{1}^I$ are zero-mean   Gaussian distributed with variance $\frac{P}{2}\frac{N-1}{N}$, we can write  \eqref{aa1.1-1} and \eqref{aa1.12-1} equivalently as 
\begin{align}
    z_{1}^R&=\textrm{sign}\left(\sqrt{\frac{P}{M}}\left( \sum_{m=1}^{M/N} (h_{1m}^R x_{m}^R-h_{1m}^I x_{m}^I) + \hat w_1^R\right) \right)\label{aa1.1-1.1} ,\\
    z_{1}^I&=\textrm{sign}\left( \sqrt{\frac{P}{M}} \left(\sum_{m=1}^{M/N} (h_{1m}^R x_{m}^I+h_{1m}^I x_{m}^R)  + \hat w_1^I\right)\right), \label{aa1.12-1.1}
\end{align}
where $\hat w_n^R$ and $\hat w_n^I$ are independent zero-mean additive white Gaussian noises both with variance 
 \begin{align}\label{eq_var1}
    \sigma_{\hat w}^2= \frac{P}{2}\frac{N-1}{N}+\frac{1}{2}=  \frac{P(N-1)+N}{2N}.
\end{align}
Now, for clarity of presentation, for a given $m$, we   represent $h_{1m}^R x_{m}^R-h_{1m}^I x_{m}^I$ and $h_{1m}^R x_{m}^I+h_{1m}^I x_{m}^R$ in \eqref{aa1.1-1.1} and \eqref{aa1.12-1.1}, respectively, in a matrix form as
\begin{align}\label{eq_max}
  \begin{bmatrix} 
h_{1m}^R x_{m}^R-h_{1m}^I x_{m}^I \\
h_{1m}^R x_{m}^I+h_{1m}^I x_{m}^R 
\end{bmatrix}
=
  \begin{bmatrix} 
h^R_{1m} & -h^I_{1m} \\
h^I_{1m} & h^R_{1m} 
\end{bmatrix}
\begin{bmatrix} 
x^R_m \\
x^I_m 
\end{bmatrix}.
\end{align}
By inserting $x^R_m$ and $x^I_m $ from (\ref{aa1})  into   (\ref{eq_max}), we obtain
\begin{align}\label{eq_max1}
\begin{bmatrix} 
h_{1m}^R x_{m}^R-h_{1m}^I x_{m}^I \\
h_{1m}^R x_{m}^I+h_{1m}^I x_{m}^R 
\end{bmatrix}
=&  \frac{1}{2}  \begin{bmatrix} 
h^R_{1m} & -h^I_{1m} \\
h^I_{1m} & h^R_{1m} 
\end{bmatrix}
  \begin{bmatrix} 
g^R_{1m} & g^I_{1m} \\
-g^I_{1m} & g^R_{1m} 
\end{bmatrix}
\begin{bmatrix} 
s^R_1 \\
s^I_1 
\end{bmatrix} \nonumber\\
=&\frac{1}{2}\begin{bmatrix} 
g^R_{1m} h^R_{1m}+g^I_{1m} h^I_{1m} &  g^I_{1m} h^R_{1m}-g^R_{1m} h^I_{1m} \\
-g^I_{1m} h^R_{1m}+g^R_{1m} h^I_{1m} &  g^R_{1m} h^R_{1m}+g^I_{1m} h^I_{1m}
\end{bmatrix}
\begin{bmatrix} 
s^R_1 \\
s^I_1
\end{bmatrix}.
\end{align}
Now, there are four cases depending on whether $g^R_{1m}=g^I_{1m}$ or $g^R_{1m}=-g^I_{1m}$ holds, and depending on whether $s^R_1=s^I_1$ or $s^R_1=-s^I_1$ holds.  If $g^R_{1m}=g^I_{1m}$ and $s^R_1=s^I_1$ hold, or if  $g^R_{1m}=-g^I_{1m}$ and $s^R_1=-s^I_1$ hold, \eqref{eq_max1} simplifies to
\begin{align}\label{eq_max2}
\begin{bmatrix} 
h_{1m}^R x_{m}^R-h_{1m}^I x_{m}^I \\
h_{1m}^R x_{m}^I+h_{1m}^I x_{m}^R 
\end{bmatrix}
= \begin{bmatrix} 
s^R_1 g^R_{1m}  h^R_{1m} \\
s^I_1 g^I_{1m}  h^I_{1m}
\end{bmatrix}.
\end{align}

If $g^R_{1m}=g^I_{1m}$ and $s^R_1=-s^I_1$ hold, or if $g^R_{1m}=-g^I_{1m}$ and $s^R_1=s^I_1$ hold,  \eqref{eq_max1} simplifies to
\begin{align}\label{eq_max2.1}
\begin{bmatrix} 
h_{1m}^R x_{m}^R-h_{1m}^I x_{m}^I \\
h_{1m}^R x_{m}^I+h_{1m}^I x_{m}^R 
\end{bmatrix}
= \begin{bmatrix} 
s^R_1 g^I_{1m}  h^I_{1m} \\
s^I_1 g^R_{1m}  h^R_{1m}
\end{bmatrix}
\end{align}
On the other hand, we have the following depending on whether the estimation is correct or not
\begin{align}\label{eq.eq}
  g^\alpha_{1m}  h^\alpha_{1m} =\left\{
  \begin{array}{rl}
  |h^\alpha_{1m}|    & \textrm{ if } h^\alpha_{1m} \textrm{ is correctly estimated}\\
  -|h^\alpha_{1m}|     & \textrm{ if } h^\alpha_{1m} \textrm{ is incorrectly estimated} 
  \end{array}
  \right.\quad \alpha\in\{R,I\}.
\end{align}
Using \eqref{eq.eq}, we can   write $h_{1m}^R x_{m}^R-h_{1m}^I x_{m}^I $ in \eqref{eq_max1} equivalently as
\begin{align}
  h_{1m}^R x_{m}^R-h_{1m}^I x_{m}^I =\left\{
  \begin{array}{rl}
  |\hat h^R_{1m}|    & \textrm{ if } h^R_{1m} \textrm{ is correctly estimated and } x_m^I=0 \\
      & \textrm{ or } h^I_{1m} \textrm{ is correctly estimated and } x_m^R=0. \\
  -|\hat h^R_{1m}|    & \textrm{ if } h^R_{1m} \textrm{ is incorrectly estimated and } x_m^I=0 \\
      & \textrm{ or } h^I_{1m} \textrm{ is incorrectly estimated and } x_m^R=0, \\
  \end{array}
  \right. 
\end{align}
where $\hat h^R_{1m}$ is a zero-mean real-valued Gaussian random variable with variance $1/2$.
Similarly, we can   write $h_{1m}^R x_{m}^R-h_{1m}^I x_{m}^I $ in \eqref{eq_max1} equivalently as
\begin{align}
  h_{1m}^R x_{m}^I+h_{1m}^I x_{m}^R  =\left\{
  \begin{array}{cc}
  |\hat h^I_{1m}|    & \textrm{ if } h^R_{1m} \textrm{ is correctly estimated and } x_m^R=0 \\
      & \textrm{ or } h^I_{1m} \textrm{ is correctly estimated and } x_m^I=0. \\
 -|\hat h^I_{1m}|    & \textrm{ if } h^R_{1m} \textrm{ is correctly estimated and } x_m^R=0 \\
      & \textrm{ or } h^I_{1m} \textrm{ is correctly estimated and } x_m^I=0, \\
  \end{array}
  \right. 
\end{align}
where $\hat h^R_{1m}$ is a zero-mean real-valued Gaussian random variable with variance $1/2$.

Without loss of generality, assume that there are $K_R$ incorrect estimates that influence $ h_{1m}^R x_{m}^I+h_{1m}^I x_{m}^R $ and  $K_I$ incorrect estimates that influence $ h_{1m}^R x_{m}^I+h_{1m}^I x_{m}^R $. Then,
 we can write (\ref{aa1.1-1}) and (\ref{aa1.12-1}) equivalently   as
\begin{align}
z_{1}^R&=\textrm{sign}\left(\sqrt{\frac{P}{M}}    s_1^R   \sum_{m=K^R+1}^{M/N } |\hat h_m^R| -\sqrt{\frac{P}{M}}    s_1^R   \sum_{m=1}^{K^R} |\hat h_m^R|+ \hat w_1^R\right) \label{aa1.1-5-1a} ,\\
    z_{1}^I&=\textrm{sign}\left(\sqrt{\frac{P}{M}}   s_1^I\sum_{m=K^I+1}^{M/N} |\hat h_m^I|-\sqrt{\frac{P}{M}}   s_1^I\sum_{m=1}^{K^I} |\hat h_m^I| + \hat w_1^I\right). \label{aa1.12-5-1a}
\end{align}

Since the received real-valued symbol at the first antenna, $z_1^R$,  is given by (\ref{aa1.1-5-1a}), $\mm{Pr}\{z_{1}^R\neq  s_{1}^R\}$ is given by 
 \begin{equation}\label{q1}
\mm{Pr}\{z_{1}^R\neq  s_{1}^R\}=  \mm{Pr}\left( \textrm{sign}\left(\sqrt{\frac{P}{M}}    s_1^R   \sum_{m=K^R+1}^{M/N } |\hat h_m^R| -\sqrt{\frac{P}{M}}    s_1^R   \sum_{m=1}^{K^R} |\hat h_m^R|+ \hat w_1^R\right)\neq s_1^R\right).
 \end{equation}

Setting $L=M/N$, \eqref{q1} can be written as
\begin{align}\label{q3}
&\mm{Pr}\{z_{1}^R\neq  s_{1}^R\} = \sum_{k=0}^L  \mm{Pr}\left(K^R=k\right) \int_{\hat h_1^R}\cdots\int_{\hat h_L^R}\nonumber\\
&\mm{Pr}\left( \textrm{sign}\left(\sqrt{\frac{P}{M}}    s_1^R   \sum_{m=k+1}^{L } |\hat h_m^R| -\sqrt{\frac{P}{M}}    s_1^R   \sum_{m=1}^{k} |\hat h_m^R|+ \hat w_n^R\right)\neq s_1^R \right)  \prod_{m=1}^{L}f(\hat h_m^R)\mm d\hat h_m^R\nonumber\\
&=\sum_{k=0}^L  \mm{Pr}\left(K^R_1=k\right) \int_{\hat h_1^R}\cdots\int_{\hat h_L^R}Q\left(\frac{\sqrt{\frac{P}{M}}\left(    \sum\limits_{m=k+1}^{L } |\hat h_m^R| -   \sum\limits_{m=1}^{k} |\hat h_m^R|\right)}{\sqrt{\frac{P(N-1)+N}{2N}}} \right) \prod_{m=1}^{L}f(\hat h_m^R)\mm d\hat h_m^R.
\end{align}
In \eqref{q3}, $\mm{Pr}\left(K^R=k\right)$ is the probability of receiving $k$ incorrect 1-bit CSI's that affect  $z_1^R$, which can be found as
\begin{equation}\label{q4}
    \mm{Pr}\left(K^R=k\right)=\binom{L}{k}p_{\epsilon}^{k}\left(1-p_{\epsilon}\right)^{L-k}, 
\end{equation}
where $p_{\epsilon}$ is given in \eqref{er2} and has been derived in Appendix \ref{app3}. 
Due to the law of large numbers, as $L\to\infty$, we have the following asymptotic equality
\begin{align}\label{q6}
  & \frac{1}{L} \left(\sqrt{\frac{P}{M}} \sum\limits_{m=k+1}^{L } |\hat h_m^R| -\sqrt{\frac{P}{M}}    \sum\limits_{m=1}^{k} |\hat h_m^R|\right) \to \frac{1}{L} (L-k) E\{|\hat h_m^R|\}-\frac{1}{L} k E\{|\hat h_m^R|\}\nonumber\\
  & =  \frac{L-2k}{L}   E\{|\hat h_m^R|\},\textrm{ as } L\to\infty .
\end{align}
 As a result of \eqref{q6}, \eqref{q3} for $L\to\infty$ simplifies to 
\begin{align}\label{q7}
 \mm{Pr}\{z_{1}^R\neq  s_{1}^R\}&\to\sum_{k=0}^L  \mm{Pr}\left(K^R_1=k\right)\nonumber\\
&\times\int_{\hat h_1^R}\cdots\int_{\hat h_L^R}Q\left(\frac{\sqrt{\frac{P}{M}}    (L-2k)    E\{|\hat h_m^R|\}}{\sqrt{\frac{P(N-1)+N}{2N}}} \right)\prod_{m=1}^{L}f(\hat h_m^R)\mm d\hat h_m^R\nonumber\\
 &= \sum_{k=0}^L  \mm{Pr}\left(K^R_1=k\right) Q\left(\frac{\sqrt{\frac{P}{M}}    (L-2k)    E\{|\hat h_m^R|\}}{\sqrt{\frac{P(N-1)+N}{2N}}} \right) .
\end{align}
Since $k\geq 0$, \eqref{q7} can be  upper bounded for $L\to\infty$  as follows
\begin{align}\label{q8}
\mm{Pr}\{z_{1}^R\neq  s_{1}^R\}&\to \sum_{k=0}^{L/2}  \mm{Pr}\left(K^R_1=k\right) Q\left(\frac{\sqrt{\frac{P}{M}}(L-2k)      E\{|\hat h_m^R|\}}{\sqrt{\frac{P(N-1)+N}{2N}}} \right)\nonumber\\
  &+\sum_{k=L/2+1}^{L}  \mm{Pr}\left(K^R_1=k\right) Q\left(\frac{\sqrt{\frac{P}{M}}(L-2k)      E\{|\hat h_m^R|\}}{\sqrt{\frac{P(N-1)+N}{2N}}} \right)\nonumber\\
  &\overset{(a)}{\leq} \sum_{k=0}^{L/2}  \mm{Pr}\left(K^R_1=k\right)Q\left(\frac{\sqrt{\frac{P}{M}}(L-2k)      E\{|\hat h_m^R|\}}{\sqrt{\frac{P(N-1)+N}{2N}}} \right)\nonumber\\
  &+\sum_{k=L/2+1}^{L}  \mm{Pr}\left(K^R_1=k\right)\nonumber\\
  &\overset{(b)}{\leq} \sum_{k=0}^{L/2}  \mm{Pr}\left(K^R_1=k\right)e^{-\beta^2(L-2k)^2/2}+\sum_{k=L/2+1}^{L}  \mm{Pr}\left(K^R_1=k\right),
\end{align}
    where 
 $(a)$ follows since $Q(x)\leq 1$, $\forall x$,    $(b)$ follows since for $x\geq 0$
\begin{equation}\label{q8.1}
    Q(x)\leq \frac{e^{-x^2}}{2} 
\end{equation}
holds, and $\beta$ is given by
\begin{equation}\label{q9}
\beta=   \frac{\sqrt{\frac{P}{M}}      E\{|\hat h_m^R|\}}{\sqrt{\frac{P(N-1)+N}{2N}}}. 
\end{equation}
Now, by substituting   \eqref{q4} into  \eqref{q8}, we obtain for $L\to\infty$
\begin{align}\label{q10}
 \mm{Pr}\left\{z^R_1\neq s^R_1\right\}&\leq  \sum_{k=0}^{L/2} \binom{L}{k} p_{\epsilon}^k(1-p_{\epsilon})^{L-k}e^{-\beta^2(L-2k)^2/2}+\sum_{k=L/2+1}^{L} \binom{L}{k} p_{\epsilon}^k(1-p_{\epsilon})^{L-k}\nonumber\\
 &\overset{(c)}{\leq}  \overbrace{\sum_{k=0}^{L/2}\mm{Pr}\left(K^R=k\right)e^{-\beta^2(L-2k)^2/2} }^{\mathcal O_1}+\overbrace{\sum_{k=L/2+1}^{L} 2^Lp_{\epsilon}^{k}(1-p_{\epsilon})^{L-k}}^{\mathcal O_2},
\end{align}
where $p_{\epsilon}$ is given in \eqref{er2} and $(c)$ follows from  $\binom{L}{k}\leq 2^L$. We can upper bound $\mathcal O_2$ as

\begin{align}\label{q10.2}
 \mathcal O_2&=   \sum_{k=L/2+1}^{L} \binom{L}{k} p_{\epsilon}^k(1-p_{\epsilon})^{L-k}< 2^L \frac{L}{2}\mm{max}\left\{ p_{\epsilon}^k(1-p_{\epsilon})^{L-k}\right\}< 2^L \frac{L}{2}p_{\epsilon}^{L/2}(1-p_{\epsilon})^{L/2}\nonumber\\
 &=\frac{L}{2}\left(2^2 p_{\epsilon}(1-p_{\epsilon})\right)^{L/2} \to 0,
\end{align}
since $4 p_{\epsilon}(1-p_{\epsilon})<1$ for $p_{\epsilon}<1/2$. To upper bound $\mathcal O_1$, we  consider two cases. Let $\mathcal E_1$ and  $\mathcal E_2$  be defined as
\begin{align}
    \mathcal E_1=\left\{k: 0\leq k\leq L/2; \; k/L\nrightarrow \frac{1}{2} \textrm{ as }L\to\infty\right\},
\end{align}
\begin{align}
    \mathcal E_2=\left\{k: 0\leq k\leq L/2; \; k/L\to \frac{1}{2} \textrm{ as }L\to\infty\right\}.
\end{align}
Then, using \eqref{q4},  $\mathcal O_1$ is upper bounded as 
\begin{align}\label{q12}
    \mathcal O_1&\leq\sum_{k\in \mathcal E_1} \mm{Pr}\left(K^R=k\right)e^{-\beta^2(L-2k)^2/2}+\sum_{k\in \mathcal E_2} 2^{L}p_{\epsilon}^k(1-p_{\epsilon})^{L-k}e^{-\beta^2(L-2k)^2/2}.
\end{align}
Now, the first sum in \eqref{q12} is upper bounded as
\begin{align}\label{q12s}
   \sum_{k\in \mathcal E_1} \mm{Pr}\left(K^R=k\right)e^{-\beta^2(L-2k)^2/2}   \leq\sum_{k\in \mathcal E_1}e^{-\beta^2L^2(1-2\frac{k}{L})^2}\to 0 \hspace{4mm}\mm{as}\hspace{4mm}L\to\infty.
\end{align}
On the other hand, the second sum in \eqref{q12} is upper bounded as 
\begin{align}\label{q13}
    &\sum_{k\in \mathcal E_2} 2^{L}p_{\epsilon}^k(1-p_{\epsilon})^{L-k}e^{-\beta^2(L-2k)^2/2}
   =  \sum_{k\in \mathcal E_2} 2^{L}p_{\epsilon}^{L\frac{k}{L}}(1-p_{\epsilon})^{L(1-k/L)}e^{-\beta^2L^2(1-2k/L)^2/2}
    \nonumber\\
    &=  \sum_{k\in \mathcal E_2} 2^{L}p_{\epsilon}^{L\frac{1}{2}}(1-p_{\epsilon})^{L(1-1/2)}e^{-\beta^2L^2(1-1)^2/2}
    \nonumber\\
     &=  |\mathcal E_2| (2^2 p_{\epsilon} (1-p_{\epsilon}))^{L/2} 
    \nonumber\\
    &\leq \frac{L}{2} (2^2 p_{\epsilon} (1-p_{\epsilon}))^{L/2} \to 0
\end{align}
 since $4 p_{\epsilon}(1-p_{\epsilon})<1$ for $p_{\epsilon}<1/2$. 
 Combining \eqref{q10.2}, \eqref{q12}, \eqref{q12s}, and \eqref{q13}, we obtain $\mm{Pr}\left\{z^R_1\neq s^R_1\right\}\leq 0$ as $L\to\infty$.
 

\section{Proof of Theorem \ref{theo.5}}\label{app.2}
\subsection{Converse}
 For the considered 1-bit quantized MIMO system, we have 
 \begin{equation}\label{eq10.0--1}
     \mm I(\bb x;\bb z|\bb G)=\mm H(\bb x|\bb G)-\mm H(\bb x|\bb z,\bb G)\leq\mm H(\bb x|\bb G)=2M,
 \end{equation}
 where the last inequality follows due to the 1-bit quantized inputs.
 Hence, the  capacity of this system cannot be larger than $2M$. In the following, we prove that \eqref{eq10.0--1} is asymptotically achievable when $N\to\infty$ and $M<\infty$ is fixed. 
 \subsection{Achievability}
 
 Using  Achievability Scheme 2, the transmitter sends $2M$ bits of information at each channel use. Hence, the rate of Achievability Scheme 2 is $2M$ bits per channel use. Now, we are only left to prove that the  symbols received at the receiver can be decoded with  probability of error vanishes   when $N\to\infty$ and $M<\infty$ is fixed.

Assume that $\bb s$ is transmitted and $\bb z$ received. From   $\bb z$,  we find $\bb{\hat s}$ using \eqref{eq_dec-s}.  Then, the probability of error   can be bounded as
\begin{align} 
   \mm{P}_{\mm e}&= {\rm Pr}\{\bb{\hat s}  \neq \bb s\}\leq \sum_{m=1}^{M}\left(\mm{Pr}\{\hat s_{m}^R\neq  s_{m}^R\}+\mm{Pr}\{\hat s_{m}^I\neq  s_{m}^I\}\right) \label{eq.204a--1}\\
   &= 2M  \mm{Pr}\{\hat s_{1}^R\neq  s_{1}^R\}, \label{eq.204b--1}
\end{align}
where \eqref{eq.204a--1} follows from the union bound and \eqref{eq.204b--1} follows from symmetry. Now, $\mm{Pr}\{\hat s_{1}^R\neq  s_{1}^R\}$  is given by 
    \begin{align}\label{eq26}
  &\mm{Pr}\{\hat s_{1}^R\neq  s_{1}^R\}= \mm{Prob}\left\{\mm{sign}\left(\sum_{n=1}^{N}  g_{n1}^R  z_{n}^R\right)\neq  s_{1}^R\right\}\nonumber\\
&=\mm{Pr}\left(\mm{sign}\left(\sum_{n=1}^{N}  g_{n1}^R\mm{sign}\left[\sqrt{\frac{P}{2M}}\sum_{m=1}^{M}(h_{nm}^R  s_{n}^R-h_{nm}^I  s_{n}^I)+w_n^R\right]\right)\neq  s_{1}^R\right)\nonumber\\
    &=\mm{Pr}\Bigg(\mm{sign}\Bigg(\sum_{n=1}^{N}  g_{n1}^R\mm{sign}\Bigg[\sqrt{\frac{P}{2M}} h_{n1}^R  s_{1}^R  \nonumber\\
    &\qquad\qquad\qquad  +\sqrt{\frac{P}{2M}}\sum_{m= 2}^{M}h_{km}^R  s_{n}^R
    -\sqrt{\frac{P}{2M}}\sum_{m= 1}^{M} h_{nm}^I  s_{n}^I 
    +w_n^R\Bigg]\Bigg)\neq  s_{1}^R\Bigg).
 \end{align}
The estimation of a given $h_{n1}^R$    is correct  with probability $1-p_{\epsilon}$, in which case $g_{n1}^R=\mm{sign}(h_{n1}^R)$ holds, and is incorrect with probability $p_{\epsilon}$, in which case $g_{n1}^R=-\mm{sign}(h_{n1}^R)$ holds. Without loss of generality, assume that  $h_{nm}^R$, for $n=1,2,...,j$ are correctly estimated and $h_{nm}^R$ for $n=j+1,2,...,N$ are incorrectly estimated, where $j$ is an RV that takes values from one to $M$. Then,  \eqref{eq26} can be written as
\begin{align}\label{ea14}
    &  \mm{Pr}\{\hat s_{1}^R\neq  s_{1}^R\} =\sum_{j=0}^{N}p_{\epsilon}^j(1-p_{\epsilon})^{N-j}\binom{N}{j}
\nonumber\\
&\times\mm{Pr}\Bigg\{\mm{sign}\Bigg(\sum_{n=1}^{N-j}\mm{sign}(h_{n1}^R) \mm{sign}\Bigg[\sqrt{\frac{P}{2M}} \left(h_{n1}^R  s_{1}^R   +\sum_{m= 2}^{M}h_{km}^R  s_{n}^R
    -\sum_{m= 1}^{M} h_{nm}^I  s_{n}^I\right)+w^R_n\Bigg]
\nonumber\\
&-\sum_{n=N-j+1}^{N}\mm{sign}(h_{n1}^R)\nonumber\\
&\times\mm{sign}\Bigg[\sqrt{\frac{P}{2M}} \left(h_{n1}^R  s_{1}^R   +\sum_{m= 2}^{M}h_{km}^R  s_{n}^R
    -\sum_{m= 1}^{M} h_{nm}^I  s_{n}^I \right)+w^R_n\Bigg]\Bigg)\neq  s_{1}^R\Bigg\}.
\end{align}
Since   $\mm{sign}(a)\mm{sign}(b)=\mm{sign}(ab)$ holds, we can  write \eqref{ea14} as
\begin{align}\label{ea15}
 &  \mm{Pr}\{\hat s_{1}^R\neq  s_{1}^R\} =\sum_{j=0}^{N}\binom{N}{j}p_{\epsilon}^j(1-p_{\epsilon})^{N-j}
\nonumber\\
&\times\mm{Pr}\Bigg\{\mm{sign}\Bigg(\sum_{n=1}^{N-j} \mm{sign}\Bigg[\sqrt{\frac{P}{2M}} |h_{n1}^R|  s_{1}^R  \nonumber\\ &+\mm{sign}(h_{n1}^R) \left\{\sqrt{\frac{P}{2M}}\left(\sum_{m= 2}^{M}h_{km}^R  s_{n}^R
    - \sum_{m= 1}^{M} h_{nm}^I  s_{n}^I\right)+ w^R_n\right\}\Bigg]
\nonumber\\
&-\sum_{n=N-j+1}^{N}  \mm{sign}\Bigg[\sqrt{\frac{P}{2M}} |h_{n1}^R|  s_{1}^R \nonumber\\
&+\mm{sign}(h_{n1}^R)\left\{ \sqrt{\frac{P}{2M}}\left(\sum_{m= 2}^{M}h_{km}^R  s_{n}^R
    -\sum_{m= 1}^{M} h_{nm}^I  s_{n}^I\right) +  w^R_n\right\}\Bigg]\Bigg)\neq  s_{1}^R\Bigg\} \nonumber\\
  & = \sum_{j=0}^{N}\binom{N}{j}p_{\epsilon}^j(1-p_{\epsilon})^{N-j}\mm{Pr}\Bigg\{\mm{sign}\Bigg(\sum_{n=1}^{N-j} \mm{sign}\Bigg[\sqrt{\frac{P}{2M}} |h_{n1}^R|  s_{1}^R   +  \hat  w_n^R\Bigg]\nonumber\\
 &-\sum_{n=N-j+1}^{N}  \mm{sign}\Bigg[\sqrt{\frac{P}{2M}} |h_{n1}^R|  s_{1}^R   +   \hat w^R_n\Bigg]\Bigg)\neq  s_{1}^R\Bigg\}, 
\end{align}
where
\begin{align}\label{ea15a}
   \hat w_n^R= \sqrt{\frac{P}{2M}}\sum_{m= 2}^{M}h_{km}^R  s_{n}^R
    + \sqrt{\frac{P}{2M}}\sum_{m= 1}^{M} h_{nm}^I  s_{n}^I+ w_n 
\end{align}
is a zero-mean Gaussian RV with variance
\begin{align}\label{ea15b}
   \sigma^2_{\hat w}= \frac{P}{2M}   \frac{2M-1}{2}
       + \frac{1}{2} .
\end{align}
Since $\mm{sign}(a)=\mm{sign}(a b)$, for any $b>0$, we can write \eqref{ea15}  as
\begin{align}\label{ea15-1}
   \mm{Pr}\{\hat s_{1}^R\neq  s_{1}^R\} &=   \sum_{j=0}^{N}\binom{N}{j}p_{\epsilon}^j(1-p_{\epsilon})^{N-j}
\mm{Pr}\Bigg\{\mm{sign}\Bigg(\frac{1}{N}\sum_{n=1}^{N-j} \mm{sign}\Bigg[\sqrt{\frac{P}{2M}} |h_{n1}^R|  s_{1}^R   +  \hat  w_n^R\Bigg]\nonumber\\
 &-\frac{1}{N} \sum_{n=N-j+1}^{N}  \mm{sign}\Bigg[\sqrt{\frac{P}{2M}} |h_{n1}^R|  s_{1}^R   +   \hat w_n^R\Bigg]\Bigg)\neq  s_{1}^R\Bigg\}.
\end{align}
Now,  the sums in \eqref{ea15-1} for $N\to\infty$  can be written  as
\begin{align}
   \frac{1}{N} \sum_{n=1}^{N-J} \mm{sign}\Bigg[\sqrt{\frac{P}{2M}} |h_{n1}^R|  s_{1}^R   +  \hat  w_n\Bigg] &=\frac{N-J}{N}\frac{1}{N-J} \sum_{n=1}^{N-J} \mm{sign}\Bigg[\sqrt{\frac{P}{2M}} |h_{n1}^R|  s_{1}^R   +  \hat  w_n\Bigg]\nonumber\\
    &\to 
        (1-\alpha_j) E\left\{\mm{sign}\Bigg[\sqrt{\frac{P}{2M}} |h_{n1}^R|  s_{1}^R   +  \hat  w_n\Bigg]\right\} ,
\end{align}
\begin{align}
 \frac{1}{N} \sum_{n=N-J+1}^{N} \mm{sign}\Bigg[\sqrt{\frac{P}{2M}} |h_{n1}^R|  s_{1}^R   +  \hat  w_n\Bigg] &=\frac{J}{N}\frac{1}{J} \sum_{n=N-J+1}^{N} \mm{sign}\Bigg[\sqrt{\frac{P}{2M}} |h_{n1}^R|  s_{1}^R   +  \hat  w_n\Bigg]\nonumber\\
    &\to  \alpha_j E\left\{\mm{sign}\Bigg[\sqrt{\frac{P}{2M}} |h_{n1}^R|  s_{1}^R   +  \hat  w_n\Bigg]\right\} , 
\end{align}
where
\begin{align}\label{eq.ex3}
    \alpha_j=\lim_{N\to\infty} \frac{j}{N}.
\end{align}
Hence, we can write \eqref{ea15-1} for $N\to\infty$ as 
\begin{align}\label{ea15-2}
   \mm{Pr}\{\hat s_{1}^R\neq  s_{1}^R\} &\to   \sum_{j=0}^{N}\binom{N}{j}p_{\epsilon}^j(1-p_{\epsilon})^{N-j}\nonumber\\
   &\mm{Pr}\Bigg\{\mm{sign}\Bigg((1-\alpha_j) E\left\{\mm{sign}\Bigg[\sqrt{\frac{P}{2M}} |h_{n1}^R|  s_{1}^R   +  \hat  w_n\Bigg]\right\}\nonumber\\
 &-\alpha_j E\left\{\mm{sign}\Bigg[\sqrt{\frac{P}{2M}} |h_{n1}^R|  s_{1}^R   +  \hat  w_n\Bigg]\right\}  \Bigg) \neq s_{1}^R\Bigg\} \nonumber\\
 &=\sum_{j=0}^{N}\binom{N}{j}p_{\epsilon}^j(1-p_{\epsilon})^{N-j}\nonumber\\
 &\times
  \mm{Pr}\left\{\mm{sign}\Bigg(E\left\{\mm{sign}\Bigg[\sqrt{\frac{P}{2M}} |h_{n1}^R|  s_{1}^R   +  \hat  w_n\Bigg]\right\} (1-2\alpha_j )   \Bigg) \neq s_{1}^R\right\} \nonumber\\
  &=\sum_{j=0}^{N}\binom{N}{j}p_{\epsilon}^j(1-p_{\epsilon})^{N-j}
  \mm{Pr}\left\{\mm{sign}\bigg(s_{1}^R  (1-2\alpha_j )  \bigg) \neq s_{1}^R\right\} \nonumber\\
  &=\sum_{j=0}^{N}\binom{N}{j}p_{\epsilon}^j(1-p_{\epsilon})^{N-j}
  \mm{Pr}\left\{\mm{sign}   (1-2\alpha_j ) \neq 1\right\} \nonumber\\
 &\overset{(a)}{=}\sum_{J=\frac{N}{2}}^{N}\binom{N}{j}p_{\epsilon}^j(1-p_{\epsilon})^{N-j}\overset{(b)}{\leq}\frac{N}{2}\binom{N}{j}p_{\epsilon}^{N/2}(1-p_{\epsilon})^{N/2}\nonumber\\
 &\overset{(c)}{\leq}\frac{N}{2}2^N p_{\epsilon}^{N/2}(1-p_{\epsilon})^{N/2}\frac{N}{2}\left(4p_{\epsilon}(1-p_{\epsilon})\right)^{N/2}\overset{(e)}{\to} 0, \textrm{as } N\to\infty,
\end{align}
where $(a)$ holds since $\mm{Pr}\left\{\mm{sign}   (1-2\alpha_j ) \neq 1\right\}=0$ for $j<\frac{N}{2}$ and $\mm{Pr}\left\{\mm{sign}(1-2\alpha_j\neq 1)\right\}=1$ for $j=\frac{N}{2}$,\dots,$N$.  $(b)$ holds since $p_{\epsilon}<\frac{1}{2}$,  $(c)$ holds since $\binom{N}{j}<2^N$, and $(e)$ holds since $4p_{\epsilon}(1-p_{\epsilon})<1$. 



\section{Probability of Erroneous 1-bit CSI Estimation}\label{app3}
In the following, we find the probability of receiving an erroneous 1-bit CSI from $K$ pilot transmissions over  the real-valued or complex-valued parts of the channel, when the symbol $x=\sqrt{P_p}$ is transmitted from one of the transmit/receive antennas. Since the distribution of the real-valued and complex-valued parts of the channel are identical, the following derivations holds for both the real-valued and complex-valued 1-bit CSI estimations.

The probability of receiving an erroneous  1-bit CSI estimate $z_n$, where $z_n\neq\mm{sign}(h_{mn})$ holds, by employing a single  pilot transmission over the  real-valued or imaginary-valued  channel $h_{mn}$ is given by 
\begin{equation}\label{er1}
\begin{split}
   &\mm{Pr}\left(z_n\neq\mm{sign}(h_{mn})\right)= \mm{Pr}\left(\mm{sign}(h_{mn})\neq\mm{sign}\left(\sqrt{P_p}h_{mn}+n_n\right)\right)\\
   &= \int_{-\infty}^{+\infty}\mm{Pr}\left(\mm{sign}(h_{mn})\neq\mm{sign}\left(\sqrt{P_p}  h_{mn}+n_n\right)\big|h_{mn}\right)f_{h_{mn}}(  h_{mn})\mm dh_{mn}\\
   & =\int_{-\infty}^{+\infty}Q\left(\sqrt{P_p}|  h_{mn}|\right)f_{h_{mn}}(  h_{mn})\mm dh_{mn}
    =\frac{1}{2\pi}\int_{-\infty}^{+\infty}\int_{\sqrt{P_p}|h_{mn}|}^{+\infty}e^{-\frac{u^2+  h^2_{mn}}{2}}\mm du\mm dh_{mn}\\
    &=\frac{1}{\pi}\int_{0}^{+\infty}\int_{\sqrt{P_p}  h_{mn}}^{+\infty}e^{-\frac{u^2+  h^2_{mn}}{2}}\mm du\mm dh_{mn}=\frac{1}{\pi}\int_{0}^{+\infty}\int_{\tan^{-1}\left(\sqrt{P_p}  h_{mn}\right)}^{\frac{\pi}{2}}e^{-\frac{r^2}{2}}r\mm dr\mm d\theta\\
    &=\frac{1}{2}-\frac{\tan^{-1}\left(\sqrt{P_p}\right)}{\pi}.
\end{split}
\end{equation}
Now, when the pilot estimation of the same channel is acquired by   $K$ pilot symbols such that the final 1-bit CSI is decided based on majority rule, an erroneous 1-bit CSI estimation of the channel occurs when half or more than half of the $K$ pilot symbols are  erroneously estimated. Since the probability of erroneous   1-bit CSI estimation from one pilot symbol was found in \eqref{er1},  the probability of receiving an erroneous 1-bit CSI from $K$ pilot symbols based on majority rule  is given by 
\begin{equation}\label{er2}
     p_{\epsilon}=\sum_{j=\frac{K}{2}}^{K}\binom{k}{j}\left(\frac{1}{2}-\frac{\arctan\left(\sqrt{P_p}\right)}{\pi}\right)^j\left(\frac{1}{2}+\frac{\arctan\left(\sqrt{P_p}\right)}{\pi}\right)^{K-j}.
\end{equation}
\end{appendices}
\bibliography{ref}

\begin{thebibliography}{10}

\bibitem{ref47}
C.-X. Wang, F.~Haider, X.~Gao, X.-H. You, Y.~Yang, D.~Yuan, H.~M. Aggoune,
  H.~Haas, S.~Fletcher, and E.~Hepsaydir, ``Cellular architecture and key
  technologies for 5g wireless communication networks,'' {\em IEEE
  communications magazine}, vol.~52, no.~2, pp.~122--130, 2014.

\bibitem{ref48}
N.~Al-Falahy and O.~Y. Alani, ``Technologies for 5g networks: Challenges and
  opportunities,'' {\em IT Professional}, vol.~19, no.~1, pp.~12--20, 2017.

\bibitem{ref6}
T.~S. Rappaport, S.~Sun, R.~Mayzus, H.~Zhao, Y.~Azar, K.~Wang, G.~N. Wong,
  J.~K. Schulz, M.~Samimi, and F.~Gutierrez, ``Millimeter wave mobile
  communications for 5g cellular: It will work!,'' {\em IEEE access}, vol.~1,
  pp.~335--349, 2013.

\bibitem{ref7}
W.~Roh, J.-Y. Seol, J.~Park, B.~Lee, J.~Lee, Y.~Kim, J.~Cho, K.~Cheun, and
  F.~Aryanfar, ``Millimeter-wave beamforming as an enabling technology for 5g
  cellular communications: Theoretical feasibility and prototype results,''
  {\em IEEE communications magazine}, vol.~52, no.~2, pp.~106--113, 2014.

\bibitem{ref49}
Huawei, ``Huawei’s massive mimo a massive hit at mwc2017.''
  \url{https://carrier.huawei.com/en/relevant-information/all-cloud-network/huaweis-massive-mimo-a-massive-hit-at-mwc2017},
  May 2017.

\bibitem{ref8}
F.~Gutierrez, S.~Agarwal, K.~Parrish, and T.~S. Rappaport, ``On-chip integrated
  antenna structures in cmos for 60 ghz wpan systems,'' {\em IEEE Journal on
  Selected Areas in Communications}, vol.~27, no.~8, 2009.

\bibitem{ref9}
W.~Hong, K.-H. Baek, Y.~Lee, Y.~Kim, and S.-T. Ko, ``Study and prototyping of
  practically large-scale mmwave antenna systems for 5g cellular devices,''
  {\em IEEE Communications Magazine}, vol.~52, no.~9, pp.~63--69, 2014.

\bibitem{ref11}
J.~Mo and R.~W. Heath, ``Capacity analysis of one-bit quantized mimo systems
  with transmitter channel state information,'' {\em IEEE transactions on
  signal processing}, vol.~63, no.~20, pp.~5498--5512, 2015.

\bibitem{ref12}
B.~Murmann, ``Adc performance survey,'' {\em CoRR, vol. abs/1404.7736},
  vol.~2014, 1997.

\bibitem{ref13}
B.~Le, T.~W. Rondeau, J.~H. Reed, and C.~W. Bostian, ``Analog-to-digital
  converters,'' {\em IEEE Signal Processing Magazine}, vol.~22, no.~6,
  pp.~69--77, 2005.

\bibitem{ref50}
S.~Ponnuru, M.~Seo, U.~Madhow, and M.~Rodwell, ``Joint mismatch and channel
  compensation for high-speed ofdm receivers with time-interleaved adcs,'' {\em
  IEEE Transactions on Communications}, vol.~58, no.~8, pp.~2391--2401, 2010.

\bibitem{ref14}
R.~H. Walden, ``Analog-to-digital converter survey and analysis,'' {\em IEEE
  Journal on selected areas in communications}, vol.~17, no.~4, pp.~539--550,
  1999.

\bibitem{ref15}
C.-C. Huang, C.-Y. Wang, and J.-T. Wu, ``A cmos 6-bit 16-gs/s time-interleaved
  adc using digital background calibration techniques,'' {\em IEEE Journal of
  Solid-State Circuits}, vol.~46, no.~4, pp.~848--858, 2011.

\bibitem{ref18}
A.~Mezghani and J.~A. Nossek, ``On ultra-wideband mimo systems with 1-bit
  quantized outputs: Performance analysis and input optimization,'' in {\em
  Information Theory, 2007. ISIT 2007. IEEE International Symposium on},
  pp.~1286--1289, IEEE, 2007.

\bibitem{ref17}
I.~D. O'Donnell and R.~W. Brodersen, ``An ultra-wideband transceiver
  architecture for low power, low rate, wireless systems,'' {\em IEEE
  Transactions on vehicular technology}, vol.~54, no.~5, pp.~1623--1631, 2005.

\bibitem{ref21}
A.~Mezghani and J.~A. Nossek, ``Capacity lower bound of mimo channels with
  output quantization and correlated noise,'' in {\em IEEE International
  Symposium on Information Theory Proceedings (ISIT)}, 2012.

\bibitem{ref44}
S.~Rini, L.~Barletta, Y.~C. Eldar, and E.~Erkip, ``A general framework for
  low-resolution receivers for mimo channels,'' {\em arXiv preprint
  arXiv:1702.08133}, 2017.

\bibitem{ref45}
A.~Khalili, S.~Rini, L.~Barletta, E.~Erkip, and Y.~C. Eldar, ``On mimo channel
  capacity with output quantization constraints,'' {\em arXiv preprint
  arXiv:1806.01803}, 2018.

\bibitem{ref23}
J.~Mo and R.~W. Heath, ``High snr capacity of millimeter wave mimo systems with
  one-bit quantization,'' in {\em Information Theory and Applications Workshop
  (ITA), 2014}, pp.~1--5, IEEE, 2014.

\bibitem{ref24}
C.~Moll{\'e}n, J.~Choi, E.~G. Larsson, and R.~W. Heath, ``Uplink performance of
  wideband massive mimo with one-bit adcs,'' {\em IEEE Transactions on Wireless
  Communications}, vol.~16, no.~1, pp.~87--100, 2017.

\bibitem{ref25}
L.~Fan, S.~Jin, C.-K. Wen, and H.~Zhang, ``Uplink achievable rate for massive
  mimo systems with low-resolution adc,'' {\em IEEE Communications Letters},
  vol.~19, no.~12, pp.~2186--2189, 2015.

\bibitem{ref19}
A.~Mezghani and J.~A. Nossek, ``Analysis of rayleigh-fading channels with 1-bit
  quantized output,'' in {\em Information Theory, 2008. ISIT 2008. IEEE
  International Symposium on}, pp.~260--264, IEEE, 2008.

\bibitem{ref20}
A.~Mezghani and J.~A. Nossek, ``Analysis of 1-bit output noncoherent fading
  channels in the low snr regime,'' in {\em Information Theory, 2009. ISIT
  2009. IEEE International Symposium on}, pp.~1080--1084, IEEE, 2009.

\bibitem{ref33}
C.~Risi, D.~Persson, and E.~G. Larsson, ``Massive mimo with 1-bit adc,'' {\em
  arXiv preprint arXiv:1404.7736}, 2014.

\bibitem{ref34}
J.~Mo, P.~Schniter, N.~G. Prelcic, and R.~W. Heath, ``Channel estimation in
  millimeter wave mimo systems with one-bit quantization,'' in {\em Signals,
  Systems and Computers, 2014 48th Asilomar Conference on}, pp.~957--961, IEEE,
  2014.

\bibitem{ref35}
O.~Dabeer and U.~Madhow, ``Channel estimation with low-precision
  analog-to-digital conversion,'' in {\em Communications (ICC), 2010 IEEE
  International Conference on}, pp.~1--6, IEEE, 2010.

\bibitem{ref36}
M.~T. Ivrlac and J.~A. Nossek, ``On mimo channel estimation with single-bit
  signal-quantization,'' in {\em ITG Smart Antenna Workshop}, 2007.

\bibitem{ref37}
J.~Choi, J.~Mo, and R.~W. Heath, ``Near maximum-likelihood detector and channel
  estimator for uplink multiuser massive mimo systems with one-bit adcs,'' {\em
  IEEE Transactions on Communications}, vol.~64, no.~5, pp.~2005--2018, 2016.

\bibitem{ref38}
A.~Mezghani, F.~Antreich, and J.~A. Nossek, ``Multiple parameter estimation
  with quantized channel output,'' in {\em Smart Antennas (WSA), 2010
  International ITG Workshop on}, pp.~143--150, IEEE, 2010.

\bibitem{ref39}
S.~Jacobsson, G.~Durisi, M.~Coldrey, U.~Gustavsson, and C.~Studer, ``One-bit
  massive mimo: Channel estimation and high-order modulations,'' in {\em
  Communication Workshop (ICCW), 2015 IEEE International Conference on},
  pp.~1304--1309, IEEE, 2015.

\bibitem{ref40}
L.~Jacques, J.~N. Laska, P.~T. Boufounos, and R.~G. Baraniuk, ``Robust 1-bit
  compressive sensing via binary stable embeddings of sparse vectors,'' {\em
  IEEE Transactions on Information Theory}, vol.~59, no.~4, pp.~2082--2102,
  2013.

\end{thebibliography}
\bibliographystyle{ieeetr}
\end{document}